\documentclass[11pt,a4paper]{article}
\usepackage[margin=3.2cm]{geometry}
\usepackage{amsmath,amssymb,amsthm,mathtools}
\usepackage{microtype,booktabs}
\usepackage[hidelinks]{hyperref}
\usepackage{authblk}

\newtheorem{theorem}{Theorem}[section]
\newtheorem{proposition}[theorem]{Proposition}
\newtheorem{lemma}[theorem]{Lemma}
\theoremstyle{remark}
\newtheorem{remark}[theorem]{Remark}

\newcommand{\dist}{d}
\newcommand{\ecc}{D}
\newcommand{\rad}{\operatorname{rad}}
\newcommand{\diam}{\operatorname{diam}}
\newcommand{\supp}{\operatorname{supp}}
\newcommand{\PoA}{\operatorname{PoA}}
\newcommand{\SC}{\operatorname{SC}}
\newcommand{\OPT}{\operatorname{OPT}}
\newcommand{\MaxGame}{\textup{\textsc{MaxGame}}}

\title{The price of anarchy in the max-distance\\
network creation game is not constant\thanks{All original results and proofs in this paper were found by OpenAI’s GPT Astra model. The author verified the results and proofs, edited the manuscript, and takes responsibility for its content.
}}
\author{Christoph Schlegel}
\affil{Flashbots, Zurich, Switzerland}
\date{}

\begin{document}
\maketitle

\begin{abstract}
At edge price $\alpha=1$, we construct an infinite family of pure Nash
equilibria of the unilateral max-distance network creation game with
$\PoA\ge2^{\sqrt{\log_2 n}-O(\log\log n)}$. Together with the known upper
bound, this gives $2^{\Theta(\sqrt{\log n})}$ along the constructed sequence
of population sizes. We subdivide every edge of the bipartite double cover
of a distance-uniform graph with large diameter constructed by Lavrov,
Loh and Messegu\'e, and let each subdivision vertex buy its two incident
edges. A distance calculation rules out every profitable unilateral
deviation. The equilibria are not strict. We also give a short proof that
the price of anarchy is constant for every polynomially vanishing edge
price.
\end{abstract}

\section{Introduction}\label{sec:intro}

In the network creation game of Fabrikant et al.~\cite{Fabrikant2003}, each
player buys incident edges and pays for both construction and distances in
the resulting graph. Demaine et al.~\cite{Demaine2012} introduced the
max-distance variant, \MaxGame, in which the usage cost is a player's
eccentricity. We show that its price of anarchy is not bounded by a constant.

\begin{theorem}\label{thm:poa}
There is a strictly increasing sequence of integers $(n_k)_{k\ge2}$ such that
\[
 \PoA(n_k,1)\ge
 2^{\sqrt{\log_2 n_k}-O(\log\log n_k)}.
\]
Consequently, along this sequence,
$\PoA(n_k,1)=2^{\Theta(\sqrt{\log n_k})}$, where the upper bound
$2^{O(\sqrt{\log n})}$ is from
Mihal\'ak and Schlegel~\cite{MihalakSchlegel2013}.
\end{theorem}

Our construction uses the distance-uniform graphs with large diameter
constructed by Lavrov, Loh and Messegu\'e~\cite{LLM}, which we call
\emph{Lavrov--Loh--Messegu\'e graphs}, or \emph{LLM graphs} for short.
Their vertices are words, and words using disjoint
sets of symbols are maximally distant. Passing to the bipartite double
cover and subdividing every edge produces a graph in which all
eccentricities are equal. Subdivision vertices buy both incident edges;
the other vertices buy nothing. For any potentially profitable deviation,
unused symbols supply a target vertex that remains sufficiently far away. The
resulting equilibrium relies on indifference at $\alpha=1$.

We give the construction and verify arbitrary deviations in
Sections~\ref{sec:construction}--\ref{sec:deviations} and thereby prove Theorem~\ref{thm:poa}. Appendix~\ref{app:hanoi} supplies
a self-contained proof of the required LLM metric facts.

\begin{table}[t]
\centering
\begin{tabular}{@{}lccccccccc@{}}
\toprule
$\alpha$
 & $0$ & & $\tfrac{2}{n-2}$ & & $n^{-\varepsilon}$ & & $19$ & & $\infty$\\
\midrule
$\PoA\le$
 & & $2$
 & & $2+\tfrac53(4^{\lceil1/\varepsilon\rceil}-1)$
 & & $2^{O(\sqrt{\log n})}$
 & & $3$ & \\
 & & {\footnotesize\cite{MihalakSchlegel2013}}
 & & {\footnotesize Prop.~\ref{prop:small}}
 & & {\footnotesize\eqref{eq:uniform}}
 & & {\footnotesize\cite{Wang2022}} & \\
\bottomrule
\end{tabular}
\caption{Upper bounds on the price of anarchy of $\MaxGame$ along the edge
price. The bound \eqref{eq:uniform} holds for every $\alpha>0$ but is weaker
than the entry shown in the outer two intervals.}\label{tab:regimes}
\end{table}

\section{Model and upper bounds}\label{sec:prelim}

There are $n\ge3$ players. Player $v$ chooses
$s_v\subseteq V\setminus\{v\}$, buying an undirected edge to each member
at price $\alpha>0$. The induced graph $G=G(s)$ gives $v$ the cost
\[
 c_v(s)=\alpha|s_v|+\ecc_G(v),
 \qquad \ecc_G(v)=\max_{w\in V}\dist_G(v,w),
\]
with infinite usage cost if $G$ is disconnected. A \textbf{pure Nash
equilibrium} admits no strictly profitable unilateral change of strategy.
Write $\SC(s)=\sum_v c_v(s)$, let $\OPT(n,\alpha)$ be the minimum social
cost, and let $\PoA(n,\alpha)$ be the largest equilibrium cost divided by
$\OPT(n,\alpha)$. Every equilibrium is connected and has each edge bought
exactly once. We therefore write $\SC(G)$ for its social cost. Finally,
$\rad(G)=\min_v\ecc_G(v)$ and $\diam(G)=\max_v\ecc_G(v)$.

For $\alpha\ge2/(n-2)$, the radius estimate of Mihal\'ak and
Schlegel~\cite[Corollary~1]{MihalakSchlegel2013} gives every equilibrium
graph $G$ the bound
\begin{equation}\label{eq:radius}
 \frac{\SC(G)}{\OPT(n,\alpha)}
 \le2+\frac{2\rad(G)}{\alpha+2}<2+\rad(G).
\end{equation}
For $\alpha<2/(n-2)$, their Theorem~2 gives $\PoA\le2$. Combining
\eqref{eq:radius} with the diameter estimate
$O\bigl((1+\alpha)4^{\sqrt{\log_2 n}}\bigr)$ from the argument of
Demaine et al.~\cite[Theorem~21]{Demaine2012} yields the uniform bound
\begin{equation}\label{eq:uniform}
 \PoA(n,\alpha)=2^{O(\sqrt{\log n})}
 \qquad(\alpha>0).
\end{equation}
For polynomially vanishing prices, their growth lemma gives a constant
bound directly.

\begin{proposition}\label{prop:small}
For $\varepsilon>0$ and $m=\lceil1/\varepsilon\rceil$,
\[
 \PoA(n,\alpha)<2+\frac53(4^m-1)
 \qquad\text{if }\alpha\le n^{-\varepsilon}.
\]
In particular, $\PoA(n,\alpha)=O(4^{1/\varepsilon})$ in this range, and
it is constant for $\alpha=O(n^{-\varepsilon})$ with fixed $\varepsilon>0$.
\end{proposition}

\begin{proof}
Put $b(t)=\min_v|\{w:\dist_G(v,w)\le t\}|$ and $r=\rad(G)$.
The growth lemma of Demaine et al.~\cite[Lemma~20]{Demaine2012} implies
\[
 r>5t\quad\Longrightarrow\quad
 b(4t+1)\ge\frac{t}{\alpha}b(t)
 \qquad(t\ge1\text{ an integer}).
\]
Set $t_0=1$ and $t_{i+1}=4t_i+1$, so
$t_i=(4^{i+1}-1)/3$. If $r>5t_{m-1}$, iterating from $b(1)\ge2$ gives
\[
 b(t_m)\ge2\alpha^{-m}\ge2n^{\varepsilon m}>n,
\]
a contradiction. Thus $r\le5t_{m-1}$, and \eqref{eq:radius} or the
small-price bound $\PoA\le2$ proves the claim. Finally,
$Cn^{-\varepsilon}\le n^{-\varepsilon/2}$ for sufficiently large $n$.
\end{proof}

Other upper bounds from the literature are summarized in
Table~\ref{tab:regimes}.
\section{The construction}\label{sec:construction}

Fix $k\ge2$ and put
\begin{equation}\label{eq:parameters}
 d=2^k-1,\qquad R=2d+2=2^{k+1},\qquad a=2k(R+1)+3.
\end{equation}

\subsection{LLM graphs, double cover and subdivision}

For an alphabet $\mathcal A$ of size $a$, the \textbf{LLM graph}
$H=H_{a,k}$, after Lavrov, Loh and Messegu\'e~\cite{LLM}, has vertices (words)
\[
 x=(x_1,\ldots,x_k)\in\mathcal A^k,
 \qquad x_i\ne x_{i+1}\quad(1\le i<k).
\]
An \textbf{adjustment} of a word $x$ changes its last symbol (subject to the condition that the last symbol should be different from the penultimate symbol). An
\textbf{involution} of a word $x$ with last two symbols $q$ and $r$, takes its longest suffix alternating
between $q$ and $r$ and exchanges $q$ and $r$ throughout this suffix, leaving
the preceding symbols unchanged. For example, if $p,q,r$ are
distinct, the involution maps
$
 (p,q,r,q,r)$ to $(p,r,q,r,q)$.

Two words are adjacent in $H$ iff one is obtained from the other
by an adjustment or by an involution. The \textbf{support} $\supp(x)$ is the set of
symbols in $x$.\footnote{These graphs arise from the generalized Hanoi game
of Lavrov, Loh and Messegu\'e. The classical three-peg Hanoi graph is obtained
by using four symbols, one designated $0$, and restricting to words whose
first symbol is nonzero.}

\begin{lemma}[LLM distances]\label{lem:hanoifacts}
Every two vertices of $H$ are at distance at most $d$. Vertices with disjoint
supports are at distance exactly $d$.
\end{lemma}

Appendix~\ref{app:hanoi} proves this lemma. Form the bipartite double cover
$B$ with vertex set $V(H)\times\{0,1\}$, joining $(x,0)$ to $(y,1)$ whenever
$x$ and $y$ are adjacent in $H$. Subdivide every edge of
$B$ once to obtain $G$. Call the vertices inherited from $B$
\textbf{original vertices} and the inserted vertices \textbf{subdivision
vertices}. We extend the notion of support to $G$ as follows. The original
vertices $(x,0)$ and $(x,1)$ have support $\supp(x)$, and the subdivision
vertex between $(x,0)$ and $(y,1)$ has support
$\supp(x)\cup\supp(y)$. Thus every vertex has support of size at most $2k$.
For $T\subseteq V(G)$, call a vertex $z$ \textbf{fresh for $T$} if
\[
 \supp(z)\cap\supp(x)=\varnothing
 \qquad\text{for every }x\in T.
\]

Orient every edge of $G$ from its subdivision endpoint toward its
original endpoint, and interpret the tail of an oriented edge as its buyer.
We call the players corresponding to subdivision vertices \textbf{buying
players} and those corresponding to original vertices \textbf{non-buying
players}.

\subsection{Distances and fresh targets}

\begin{lemma}[Parity distances]\label{lem:parity}
In $B$, the distance between $(x,\eta)$ and $(y,\theta)$ is at most $d+1$
if $\eta=\theta$, and at most $d$ otherwise. Equality holds in both cases
when $\supp(x)\cap\supp(y)=\varnothing$.
\end{lemma}

\begin{proof}
Walk in $H$ lift to $B$ with endpoint bits determined by their parity.
Thus, we want to find an odd $x$-$y$ walk of length at most $d$ and an even $x$-$y$ walk of length at most $d+1$.
Take a shortest $x$--$y$ path, of length $\ell\le d$.
If it contains an adjustment edge, its endpoints have a common prefix
of length $k-1$. Choose a third last symbol different from both endpoint
symbols and from the last symbol of this prefix; this is possible since
$a\ge4$. Replacing the edge by the two adjustments through this word
gives a walk of length $\ell+1$. Since $d$ is odd, these two lengths
satisfy the required bounds.

Otherwise every edge of the path is an involution edge. Applying the
involution twice restores the word, so $\ell\le1$.
Keeping the first $k-1$ coordinates fixed, cycle the last coordinate
through two other symbols different from $x_{k-1}$ and then back to $x_k$.
Prepending these three adjustments gives a walk of length $\ell+3$.
The two walks have opposite parity, with odd length at most $3\le d$
and even length at most $4\le d+1$.

For disjoint supports, Lemma~\ref{lem:hanoifacts} gives a lower bound
of $d$ for every projected walk, raised to $d+1$ when even parity is required.
\end{proof}
Immediately from the previous lemma we get the following:
\begin{lemma}[Distances in $G$]\label{lem:table}
The distances between two vertices $u,v$ in $G$ is
\begin{enumerate}
\item $d_G(u,v)\leq R$ for original vertices with the same bit, 
\item  $d_G(u,v)\leq R-2$ for orginal vertices with opposite bits,
\item  $d_G(u,v)\leq R-1$ if $u$ is an original vertex and $v$ is a subdivision vertex,
\item $d_G(u,v)\leq R$ for subdivision vertices.
\end{enumerate}
If $u$ and $v$ have disjoint support the bounds hold with equality.
\end{lemma}

\begin{proof}
The first two statement follow by Lemma~\ref{lem:parity} noting that
distances between original vertices double after subdividing edges. A distance involving
subdivision vertices is obtained by minimising over their original neighbours
and adding one (subdivision to original vertex) or two (subdivision to subdivision vertex). Each subdivision vertex has one
original neighbour of each bit, so an opposite-bit pair attains the minimum for
disjoint supports and supplies the required upper bound otherwise.
\end{proof}

\begin{lemma}[Fresh targets]\label{lem:fresh}
For every $T\subseteq V(G)$ with $|T|\le R+1$, there are original vertices of
both bits and a subdivision vertex that are fresh for $T$. Every vertex of
$G$ has eccentricity $R$.
\end{lemma}

\begin{proof}
The union of the supports in $T$ has size at most $2k(R+1)<a-2$.
Choose unused symbols $p,q$, a word $z$ alternating between them, and its
involution $z'$. The vertices $(z,0)$, $(z,1)$ and the subdivision vertex
between $(z,0)$ and $(z',1)$ are the required vertices.

Applying the previous observation to a singleton, we find for each vertex $v$ in $G$ three vertices, two original vertices of different bits and a subdivision vertex, which have disjoint support from $v$. By Lemma~\ref{lem:table}, this implies that $v$ has eccentricity of $R$.

\end{proof}

\section{Equilibrium and price of anarchy}\label{sec:deviations}

\begin{theorem}\label{thm:equilibrium}
The strategy profile whose induced graph is $G$ with edge ownership
given by the above orientation is a pure Nash equilibrium at $\alpha=1$.
\end{theorem}

\begin{proof}
All eccentricities are $R$ by Lemma~\ref{lem:fresh}.
Fix a deviating player $u$, let $S$ be the deviation strategy and $t:=|S|$. 
In the following we assume $1\leq t\leq R$ as other deviations are (trivially) not profitable. 

\textbf{Non-buying player.}
A non-buying player's cost under the original profile is $R$. Choose an original vertex $z$ of the same
bit as $u$, fresh for $S\cup\{u\}$. Its old distance from $u$ is $R$;
its old distance from every $x\in S$ is at least $R-2$. The new distance
from $u$ to $z$ is therefore at least
\[
 \min\{R,\,1+\min_{x\in S}\dist_G(x,z)\}\ge R-1.
\]
The new cost is at least $t+R-1\ge R$.

\textbf{Buying player.}
A buying player's cost under the original profile is $R+2$.
If $t=1$, write $S=\{x\}$. Choose a vertex fresh for $\{x,u\}$ and of the same type (original or subdivision) as $x$, as well as
of the same bit if $x$ is original. Its old distance from $x$ is $R$, so
the new eccentricity of $u$ is at least $R+1$ and $u$'s cost under the deviation is at least
$R+2$.
 If $t\ge2$, choose a subdivision vertex fresh for $S\cup\{u\}$. Its old distance from
every $x\in S$ is at least $R-1$, so the new eccentricity of $u$ is at
least $R$ and its cost at least $t+R\ge R+2$.
\end{proof}

\begin{remark}\label{rem:whyparity}
These equilibria are not strict. A non-buying player can buy one edge to any
opposite-bit original vertex and reduce its eccentricity from $R$ to $R-1$:
Lemma~\ref{lem:table} gives the upper bound, and the preceding proof gives
the lower bound. At unit price the player is indifferent; below unit
price the same deviation is profitable.
\end{remark}

\begin{proof}[Proof of Theorem~\ref{thm:poa}]
The graph $H$ has $v=a(a-1)^{k-1}$ vertices, each with $a-2$ adjustment
neighbours and one involution neighbour. Thus $B$ has $2v$ vertices and
$(a-1)v$ edges, and $G$ has
\[
 n_k=(a+1)v=(a+1)a(a-1)^{k-1}
 \quad\text{vertices and}\quad 2(a-1)v\text{ edges}.
\]
Its equilibrium social cost is $n_kR+2(a-1)v$. A star costs $3n_k-2$,
so
\[
 \PoA(n_k,1)\ge
 \frac{n_kR+2(a-1)v}{3n_k-2}>\frac{R}{3}.
\]
Since $\log_2 a=k+\log_2 k+O(1)$,
\[
 \log_2 n_k=k^2+O(k\log k),\qquad
 \log_2(R/3)=k+O(1)
 \ge\sqrt{\log_2 n_k}-O(\log\log n_k).
\]
This proves the lower bound; \eqref{eq:uniform} supplies the upper bound
for the final assertion.
\end{proof}

\section{Open problems}\label{sec:discussion}

The first question is which other edge prices admit unbounded price
of anarchy. Remark~\ref{rem:whyparity} rules out this ownership profile
for $\alpha<1$.

Second, the leading constant in the exponent remains open. Our lower
bound has coefficient $1$ in front of $\sqrt{\log_2 n}$, whereas the
upper bound used in \eqref{eq:uniform} has coefficient $2$.
The general diameter bounds for distance-uniform graphs depend also on
the exceptional fraction~\cite{LLM}; they do not by themselves determine
the sharp constant for equilibrium graphs.

Finally, and most importantly: can a non-constant lower bound on the Price of Anarchy  for the sum game  be established at
$\alpha=Cn$ for a constant $C>0$? The existing upper bounds on the Price of Anarchy in the sum game match those in the max game, after scaling the edge price by $n$: in general an upper bound of $2^{\mathcal{O}(\sqrt{\log n})}$ holds for both games~\cite{MihalakSchlegel2013}, the price of anarchy is constant for sub-linear cost $n^{1-\epsilon}$ in the sum game \cite{Demaine2012} and sub-constant cost $n^{-\epsilon}$ (Proposition~\ref{prop:small}) in the max game, and there exists constants $C_{sum},C_{max}>0$ so that for cost at least $C_{sum}n$ in the sumgame and cost at least $C_{max}$ in the max game all equilibria are trees which have a constant price of anarchy.\footnote{The current state of the art is $C_{sum}=2$ due to \cite{dippel2023} and $C_{max}=19$ due to \cite{Wang2022}. However, these are not established to be tight thresholds for tree equilibria.}
Dividing costs in the sum game by $n$ makes the
edge price $C$ and replaces the distance sum by an average, not a
maximum. Thus the price scales correspond, but the equilibrium
conditions do not in general. The construction in the present paper and related constructions do not yield equilibria for the sum game.

\appendix
\section{LLM metric facts}\label{app:hanoi}

We prove Lemma~\ref{lem:hanoifacts}, taking $H_{a,1}$ to be the complete
graph on $\mathcal A$.

\textbf{Upper bound.}
Inductively, any $x$ can be transformed into $y$ in at most $2^k-1$ moves
while its first coordinate stays in $\{x_1,y_1\}$. This is immediate for
$k=1$. If $x_1=y_1$, apply induction to the last $k-1$ coordinates.
Their first coordinate stays in $\{x_2,y_2\}$ and hence differs from
$x_1$. Every move therefore lifts to the full word: an involution
extending into the fixed coordinate would make the new second
coordinate equal to it. This takes at most $2^{k-1}-1$ moves.
If $x_1\ne y_1$, first reach the alternating word on $x_1,y_1$ with
first coordinate $x_1$, involute the whole word, and then reach $y$
with first coordinate $y_1$ fixed. The two fixed-coordinate parts use
at most $2^{k-1}-1$ moves each, giving $2^k-1$ in total.

\textbf{Lower bound.}
Induct on $k$, again with immediate base case. Choose a disjoint-support
pair of minimum distance and a shortest path between them. Its first
and last moves are adjustments, since an involution preserves support
and could otherwise be removed to give a closer disjoint pair.
Consecutive adjustments combine and consecutive involutions cancel,
so the moves alternate. Delete the last coordinate along this path.
Adjustments become stationary; an involution becomes an adjustment if
its alternating suffix has length two, and an involution otherwise.
The resulting walk joins disjoint-support words in $H_{a,k-1}$, so
there were at least $2^{k-1}-1$ involutions and $2^{k-1}$ adjustments.
The minimum distance of a disjoint pair is therefore at least $2^k-1$.


\begin{thebibliography}{9}

\bibitem{Demaine2012}
E.~D. Demaine, M.~Hajiaghayi, H.~Mahini, and M.~Zadimoghaddam.
\newblock The price of anarchy in network creation games.
\newblock \emph{ACM Transactions on Algorithms}, 8(2), 2012.
\newblock Conference version: PODC 2007.

\bibitem{Fabrikant2003}
A.~Fabrikant, A.~Luthra, E.~Maneva, C.~H. Papadimitriou, and S.~Shenker.
\newblock On a network creation game.
\newblock In \emph{22nd ACM Symposium on Principles of Distributed Computing
(PODC)}, pages 347--351, 2003.

\bibitem{LLM}
M.~Lavrov, P.-S.~Loh, and A.~Messegu\'e.
\newblock Distance-uniform graphs with large diameter.
\newblock arXiv:1703.01477, 2017.

\bibitem{MihalakSchlegel2013}
M.~Mihal\'ak and J.~C. Schlegel.
\newblock The price of anarchy in network creation games is (mostly) constant.
\newblock \emph{Theory of Computing Systems}, 53(1):53--72, 2013.

\bibitem{Wang2022}
Q.~Wang.
\newblock On tree equilibria in max-distance network creation games.
\newblock In \emph{15th International Symposium on Algorithmic Game Theory
(SAGT)}, pages 293--310, 2022. arXiv:2106.15961v3.

\bibitem{dippel2023}
J. Dippel and A. Vetta
\newblock One n remains to settle the tree conjecture
\newblock In \emph{41st International Symposium on Theoretical Aspects of Computer Science (STACS)}, pages 481--496, 2024.


\end{thebibliography}
\end{document}